\documentclass[twocolumn,prl,secnumarabic,nobibnotes,superscriptaddress,aps,10pt]{revtex4-2}

\usepackage[utf8]{inputenc}
\usepackage{CJK}
\usepackage{graphicx}
\usepackage{dcolumn}
\usepackage{bm}
\usepackage{mathtools}
\usepackage{amssymb,amsmath,amsfonts,dsfont, amsthm,pifont}
\theoremstyle{definition}
\newtheorem{lemma}{Lemma}
\newtheorem{theorem}{Theorem}

\newtheorem{definition}{Definition}
\newtheorem{assumption}{Assumption}

\newtheorem{corollary}[theorem]{Corollary}
\usepackage{multirow}
\usepackage{makecell}
\usepackage{hyperref}
\usepackage[capitalize]{cleveref}
\usepackage{xcolor}
\usepackage[ruled, lined, linesnumbered, commentsnumbered, longend]{algorithm2e}
\usepackage{tikz}
\usepackage{makecell}
\setcellgapes{2pt}

\usetikzlibrary{arrows.meta, positioning}

\newcommand{\ud}{\,\mathrm{d}}

\newcommand{\DeptMath}{Department of Mathematics, University of California, Berkeley, CA 94720, USA}
\newcommand{\LBLMath}{Applied Mathematics and Computational Research Division, Lawrence Berkeley National Laboratory, Berkeley, CA 94720, USA}
\newcommand{\Michigan}{Department of Mathematics, University of Michigan, Ann Arbor, MI 48109, USA}
\newcommand{\Pennstate}{Department of Mathematics, Pennsylvania State University, University Park, PA 16802, USA}
\newcommand{\FlatironCCQ}{Center for Computational Quantum Physics, Flatiron Institute, New York, NY 10010, USA}
\newcommand{\FlatironCCM}{Center for Computational Mathematics, Flatiron Institute, New York, NY 10010, USA}
\begin{document}


\title{Provably Efficient Long-Time Exponential Decompositions of Non-Markovian Gaussian Baths}

\author{Zhen Huang}
\affiliation{\DeptMath}
\author{Zhiyan Ding}
\affiliation{\Michigan}
\author{Ke Wang}
\affiliation{\Michigan}
\author{Jason Kaye}
\affiliation{\FlatironCCQ}
\affiliation{\FlatironCCM}
\author{Xiantao Li}
\affiliation{\Pennstate}
\author{Lin Lin}
\thanks{linlin@math.berkeley.edu}
\affiliation{\DeptMath}
\affiliation{\LBLMath}

\date{\today}

\begin{abstract}
Gaussian baths are widely used to model non-Markovian environments, yet the cost of accurate simulation at long times remains poorly understood, especially when spectral densities exhibit nonanalytic behavior as in a range of realistic models. We rigorously bound the complexity of representing bath correlation functions on a time interval $[0,T]$ by sums of complex exponentials, as employed in recent variants of pseudomode and hierarchical equations of motion methods. These bounds make explicit the dependence on the maximal simulation time $T$, inverse temperature $\beta$, and the type and strength of singularities in an effective spectral density. For a broad class of spectral densities, the required number of exponentials is bounded independently of $T$, achieving time-uniform complexity. The $T$-dependence emerges only as polylogarithmic factors for spectral densities with strong singularities, such as step discontinuities and inverse power-law divergences. The temperature dependence is mild for bosonic environments and disappears entirely for fermionic environments. Thus, the true bottleneck for long-time simulation is not the simulation duration itself, but rather the presence of sharp nonanalytic features in the bath spectrum. Our results are instructive both for long-time simulation of non-Markovian open quantum systems, as well as for Markovian embeddings of classical generalized Langevin equations with memory kernels.
\end{abstract}

\maketitle

\emph{Introduction.--}
The model of a quantum system linearly coupled to a Gaussian environment is widely used across quantum optics \cite{RivasHuelga2012,Forn-DiazLamataRicoKonoetal2019}, condensed matter physics \cite{GeorgesKotliarKrauthetal1996,HaertleCohenReichmanetal2013}, and chemical and biological physics \cite{MukamelAbramavicius2004,TanakaTanimura2009}. While Markovian approximations capture short-time dynamics and weak system-bath coupling, many relevant regimes extend beyond weak coupling and are intrinsically non-Markovian. Such non-Markovian effects can be captured by \emph{pseudomodes}~\cite{Garraway1997,Dorda2014, Dorda2015,WoodsPlenio2016,RosenbachCerrilloHuelgaetal2016,MascherpaSmirneHuelgaetal2017, TamascelliSmirneHuelgaPlenio2018,TamascelliSmirneLimetal2019,
NusselerDhandHuelgaPlenio2020, Tanimura2020,XuYanShietal2022, ParkHuangZhuetal2024, ThoennissVilkoviskiyAbanin2024, lednev2024lindblad, HuangParkChanLin2026}, which are auxiliary degrees of freedom that reproduce the bath correlation functions (BCFs) and enable efficient simulation schemes.

The computational scaling of different pseudomode constructions can vary widely, particularly for simulations at long simulation times ($T$). Hamiltonian-based pseudomode constructions such as TEDOPA suffer from polynomial scaling with $T$ \cite{DeVegaSchollwoeckWolf2015,WoodsPlenio2016,RosenbachCerrilloHuelgaetal2016,TamascelliSmirneLimetal2019,NusselerDhandHuelgaPlenio2020} which has become an obstruction towards practical long-time simulations.  Similarly, hierarchical equations of motion (HEOM) methods \cite{TanakaTanimura2009,Tanimura2020} also suffer from similar prohibitive scalings in $T$. Recent advances such as quasi-Lindblad pseudomodes \cite{ParkHuangZhuetal2024,ThoennissVilkoviskiyAbanin2024}, coupled-Lindblad pseudomodes~\cite{lednev2024lindblad,HuangParkChanLin2026}, and free-pole HEOM (FP-HEOM) \cite{XuYanShietal2022} represent BCFs as optimized sums of exponentials with complex weights and poles. This enables efficient long-time simulations across a wide range of physical systems.

Despite their practical success, rigorous complexity analysis has only appeared recently: the sharpest available bounds in the analytic setting were obtained in \cite{ThoennissVilkoviskiyAbanin2024} and rely on strong analyticity assumptions on the spectral density. Many physically relevant settings do not satisfy such assumptions. For instance, the density of states of periodic systems generically exhibits van Hove singularities~\cite{VanHove1953}: inverse power-law divergences at band edges in one dimension, logarithmic divergences in two dimensions, and a square-root cusp in three dimensions.

Even within the analytic setting, existing upper bounds are not tight. As an illustrative example, we consider the spin--boson model with an Ohmic spectral density \(J(\omega)=\omega e^{-\omega}\) at zero temperature. Figure~\ref{fig:n_vs_t} shows the number of bath modes \(N\) required to achieve a fixed accuracy \(\varepsilon=0.01\) in the $L^1$ norm for reproducing the bath correlation function over the time interval $[0,T]$, plotted as a function of the maximal simulation time $T$. For comparison, we include three reference scalings, $N\propto T$, $N\propto \log T$, and $N\propto (\log T)^2$ (the upper bound established for this setting in \cite{ThoennissVilkoviskiyAbanin2024}). Numerical results (for details, see \cite[\cref{sec:numerics}]{supp}) suggest that the required number of bath modes can be asymptotically independent of $T$, which would lead to a significant saving in simulation cost.

\begin{figure}[h]
    \includegraphics[width=3.25in]{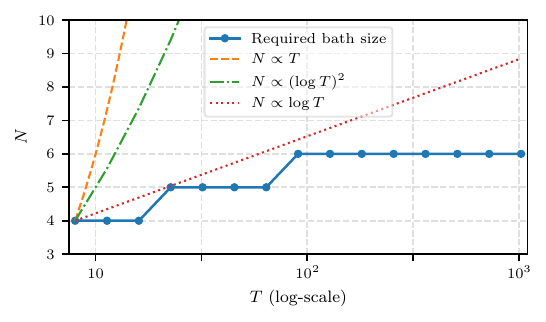}
    \caption{Number of bath modes $N$ required to represent the bath of an Ohmic spin-boson model at zero temperature with fixed accuracy $\varepsilon = 0.01$ in $L^1$ norm, as a function of the maximum simulation time $T$.}
    \label{fig:n_vs_t}
\end{figure}

\begin{table*}[ht]
\begin{tabular}{|c|cl|cl|cl|cl|}
\hline
 &
  \multicolumn{2}{c|}{\begin{tabular}[c]{@{}c@{}}Super-Ohmic \\ Bosonic bath\end{tabular}} &
  \multicolumn{2}{c|}{\begin{tabular}[c]{@{}c@{}}Sub-Ohmic \\ Bosonic bath\end{tabular}} &
  \multicolumn{2}{c|}{\begin{tabular}[c]{@{}c@{}}Gapped \\ Fermionic bath\end{tabular}} &
  \multicolumn{2}{c|}{\begin{tabular}[c]{@{}c@{}}Gapless \\ Fermionic bath\end{tabular}} \\ \hline
\begin{tabular}[c]{@{}c@{}}Zero\\ Temperature\end{tabular} &
  \multicolumn{2}{c|}{$O\left(\log^2\left(\frac 1{\omega_c\varepsilon}\right)\right) $} &
  \multicolumn{2}{c|}{$O\left(\log^2\left(\frac 1{\omega_c\varepsilon} \right) \right)$} &
  \multicolumn{2}{c|}{\multirow{3}{*}{$O\left(\log^2\left(\frac{1}{\omega_c\varepsilon}\right)\right)$}} &
  \multicolumn{2}{c|}{\multirow{3}{*}{$O\!\left(\log\!\left( \frac{\log(\omega_c T)}{\omega_c\varepsilon}\right)\,\log\!\left(\frac{ T}{\varepsilon}\right)\right)$}} \\ \cline{1-5}
\begin{tabular}[c]{@{}c@{}}Finite\\ Temperature\end{tabular} &
    \multicolumn{2}{c|}{$ {O\left( \log^2\left((1+\frac{1}{\beta\omega_c})\cdot\frac {1}{\omega_c\varepsilon} \right)  \right)}$} &
    \multicolumn{2}{c|}{$ {O\left(\log^2\left(\frac 1{\omega_c\varepsilon} \right) + \log^2\left(\frac{1}{\beta\omega_c}\cdot\frac{T}{\varepsilon}\right)\right)}$} &
  \multicolumn{2}{c|}{} &
  \multicolumn{2}{c|}{} \\ \hline
\end{tabular}
\caption{Summary of the complexity bounds for representing bath correlation functions of various physically relevant systems, with $L^1$ accuracy $\varepsilon$ on $[0,T]$. Here $\omega_c$ is a characteristic frequency scale of the bath (for example, see \cref{eq:Ohmic} for its definition in the Ohmic-family spectral densities). The bounds are independent of $T$ for super-Ohmic bosonic baths and gapped fermionic baths at all temperatures and for sub-Ohmic bosonic baths at zero temperature, and depend on $T$ only through polylogarithmic factors for finite temperature sub-Ohmic bosonic baths and gapless fermionic baths.  {For simplicity, we have omitted the results for Ohmic bosonic baths in this table, which are $O\left(\log^2(\frac{1}{\omega_c\varepsilon}) + \log(\frac{\log(\omega_c T)}{\beta \omega_c^2 \varepsilon} )\log(\frac{T}{\beta \omega_c \varepsilon})\right)$.}
}
\label{tab:summary}
\end{table*}

We note that the availability of compact sum-of-exponentials (SOE) representations of imaginary-time functions with provably logarithmic scaling in inverse temperature has enabled many algorithmic developments in equilibrium quantum many-body calculations. Approaches such as the discrete Lehmann representation (DLR) \cite{kaye2022discrete} and various pole-fitting \cite{NakatsukasaSeteTrefethen2018,Mejuto-ZaeraZepeda-NunezLindseyetal2020,shinaoka21,HuangGullLin2023,Ying2022} schemes have led to efficient algorithms for  
diagrammatics \cite{KayeHuangStrandGolez2024,HuangDenisStrandKaye2025,GazizovaZhangGullLeBlanc2024,gazizova25}, analytic continuation \cite{HuangGullLin2023,zhang24,ZhangYuGull2024,ying22,Ying2022}, and dynamical mean-field theory (DMFT) calculations \cite{georges96,labollita25,sheng23}, in which many impurity solvers, such as those based on exact diagonalization \cite{caffarel94,liebsch11}, diagrammatics \cite{KayeHuangStrandGolez2024,HuangDenisStrandKaye2025}, matrix product states \cite{wolf15}, and neural network learning \cite{valenti26} all
require a compact bath discretization as input. A natural and frequently raised question is the efficiency of similarly compact representations in nonequilibrium real-time calculations (such as \cite{XuYanShietal2022, ParkHuangZhuetal2024, HuangParkChanLin2026, ZhangErpenbeckYuetal2025,PaprotzkiEckstein2025}). This work addresses that question by establishing rigorous complexity bounds.

In this Letter, we present a unified, rigorous framework for determining the simulation complexity of real-time Gaussian environments.
Our approach advances prior analyses in various ways. We establish complexity bounds that remain valid in the presence of spectral singularities (see \cref{tab:summary}). For many physically relevant regimes, we prove how the number of bath modes $N$ required scales with maximum simulation time $T$. For fixed target $L^1$ accuracy of the bath correlation functions on $[0,T]$, the long-time scaling depends on the singularity type: $N$ is independent of $T$ for mild singularities (e.g.\ square-root); $N=O(\log T\,\log\log T)$ for step discontinuities or logarithmic singularities; and, in the worst case, $N=O((\log T)^2)$ for inverse power-law divergences. 

\emph{Setup.--} A Gaussian environment is described by its spectral density $J(\omega)$, which is a non-negative, piecewise smooth $L^1$ function that is compactly supported or exponentially decaying. The lesser and greater bath correlation functions (for fermionic problems, they are also known as hybridization functions) are:
\begin{equation}
    \begin{aligned}
        \Delta^<(t)&=\int_{-\infty}^\infty J(\omega)f(\omega-\mu)\,\mathrm e^{-\mathrm i\omega t}\ud\omega,\\
        \Delta^>(t)&=\int_{-\infty}^\infty J(\omega)\bigl[1\pm f(\omega-\mu)\bigr]\,\mathrm e^{-\mathrm i\omega t}\ud\omega.
    \end{aligned}
\end{equation}
Here $f(\omega)$ is the Bose--Einstein function $f_{\text{BE}}(\omega) = \frac{1}{e^{\beta \omega}-1}$ for bosonic baths, and the Fermi--Dirac function $f_{\text{FD}}(\omega) = \frac{1}{e^{\beta \omega}+1}$ for fermionic baths; the sign $+$ corresponds to bosons and $-$ to fermions. $\beta$ is the inverse temperature and $\mu$ is the chemical potential. For bosonic problems, $\mu$ should lie below the bottom of the support of $J(\omega)$ (or, if $\mu$ coincides with a band edge, that $J(\omega)$ vanishes sufficiently fast there) to avoid a divergent occupation; it is conventional to set $\mu=0$ and take $J$ supported on $[0,\infty)$.

For bosonic baths we work with the total BCF
\begin{equation}\label{eqn:separation}
    \Delta(t)=\Delta^>(t)+\Delta^<(-t)=\int_{-\infty}^\infty J_{\text{eff}}(\omega)\,\mathrm e^{-\mathrm i\omega t}\ud\omega,
\end{equation}
with $J_{\text{eff}}(\omega)=\operatorname{sign}(\omega)\frac{J(|\omega|)}{1-\mathrm e^{-\beta\omega}}$.
For fermions, set $J_{\text{eff}}^<(\omega)=J(\omega)f_{\mathrm{FD}}(\omega-\mu)$ and $J_{\text{eff}}^>(\omega)=J(\omega)[1-f_{\mathrm{FD}}(\omega-\mu)]$ so that $\Delta^{>,<}(t)=\int_{-\infty}^\infty J_{\text{eff}}^{>,<}(\omega)\,\mathrm e^{-\mathrm i\omega t}\ud\omega$. Here $f_{\mathrm{FD}}$ is the Fermi--Dirac function. 
In what follows, we will omit the $>,<$ superscript since the analysis holds for both cases.

To find a compact representation of the bath, our goal is to approximate $\Delta(t)$ by a sum of exponentials:
\begin{equation}
        \Delta(t) = \sum_{j=1}^{N} c_j e^{-\mathrm i z_j t} +\delta(t),  \quad t\in[0,T],
        \label{eq:SOE}
\end{equation}
such that the error $\delta(t)$ is small. We say the approximation on $[0,T]$ has $\varepsilon$ accuracy in $L^1$ norm if $\int_0^T |\delta(t)| \,\ud t \le \varepsilon$. According to \cite{MascherpaSmirneHuelgaetal2017,HuangLinParkZhu2024}, controlling the $L^1$ bound allows us to control errors of any bounded system observables. Thus the problem of simulating Gaussian environments reduces to approximating BCFs by sums of exponentials with controlled errors.

\emph{Conformal mapping and exponentially clustering poles.}
Given $J_{\text{eff}}(\omega)$, we will treat the frequency domain $\omega>\mu$ and $\omega<\mu$ separately. If $J_{\text{eff}}(\omega)$ is not analytic globally but only piecewise analytic, we will also consider each analytic segment separately.
 The overall complexity is then the sum of the complexities for each analytic segment, which is dominated by the segment with the strongest singularity. In what follows, we use the terms {\emph{exponents}}, {\emph{poles}}, and {\emph{quadrature nodes}} interchangeably, as they all refer to the same quantities $z_j$'s defined in \cref{eq:SOE}.

The key idea of our analysis is illustrated in \cref{fig:contour}. For each smooth segment of $J_{\text{eff}}(\omega)$, we assume it admits a holomorphic extension to a domain $\Omega$ in the lower half-plane, with a semi-circular–like shape as illustrated by the gray region in \cref{fig:contour}(a). (the gray region in \cref{fig:contour}(a)). We note that this is a much weaker analyticity assumption compared to previous works \cite{ThoennissVilkoviskiyAbanin2024} (see \cref{fig:contour}(c) for comparison). To resolve possible endpoint singularities of each analytic segment, we place quadrature nodes $z_j$'s in \cref{eq:SOE} (illustrated as triangles in \cref{fig:contour}(a)) in the lower half-plane that are exponentially clustered toward the segment endpoints. Such choice of $\{z_j\}$ can be viewed as defining a nonuniform quadrature rule for integration along a contour deformed into the lower half-plane. In the Support Information (SM) \cite[\cref{sec:mainproof}]{supp}, we show that such a choice achieves the desired accuracy using a conformal map that sends the semi-elliptic region to a strip in the lower complex plane, in which the corresponding quadrature nodes are uniform (see \cref{fig:contour}(b)).Moreover, the analytic region guarantees that there is an $O(\frac{1}{\beta})$ distance from the poles of the Fermi-Dirac functions, which are the Matsubara frequencies $\nu_n = \mu - \mathrm i \frac{\pi (2n+1)}{\beta}$. This property is crucial in proving the asymptotic $\beta$-independence of the required number of poles for fermionic baths.

As noted earlier, the complexity (i.e., the number of exponentially clustered poles required near each singularity) depends on the strength of the singularity.
This dependence is rigorously formulated in SM \cite[Corollary \ref{cor:main}]{supp} and in \cite[\cref{eq:main}]{supp}, which we summarize below. In essence, the required number of poles scales logarithmically with a prefactor related to the $L^1$ norm of the bath correlation function (see \cite[\cref{sec:mainproof}]{supp}).
For weak singularities, such as $J_{\text{eff}}(\omega)\sim |\omega-\omega_0|^{\alpha}$ with $\alpha>0$, its Fourier transform decays faster than $1/t$  and  is therefore $L^1$ integrable on $[0,\infty)$, leading to a $T$-independent complexity: $O(\log^2(1/\omega_c\varepsilon))$. Here $\omega_c$ is the characteristic frequency scale of the bath.
In contrast, for a stronger singularity such as an inverse power-law behavior $J_{\text{eff}}(\omega)\sim |\omega-\omega_0|^{-\alpha}$ with $0<\alpha<1$, the hybridization function decays slower than $1/t$, and the complexity is polylogarithmic in $T$: $O(\log^2(T/\varepsilon))$. For intermediate singularities such as step or logarithmic singularities, the complexity is $O\left(\log\left(\frac{\log(\omega_cT)}{\omega_c\varepsilon}\right)\log\left(\frac{T}{\varepsilon}\right)\right)$. We refer to SM \cite[\cref{sec:setup}]{supp} for precise definitions of singularity types and \cite[\cref{sec:mainproof}]{supp} for detailed rigorous proofs.

\emph{$\log^2 T$-dependence in a pre-asymptotic regime.}  We focus on the asymptotic scaling of the required number of bath modes $N$ with respect to the maximum simulation time $T$. For mild singularities $J_{\text{eff}}(\omega)\sim |\omega-\omega_0|^{\alpha}$ with $\alpha>0$, the asymptotic bound is $T$-independent because the $L^1$ norm of the tail is uniformly bounded as $T\to\infty$. Therefore, once the approximation is made sufficiently accurate on $[0,T_*]$, where the tail $\int_{T_*}^{\infty}|\Delta(t)|\,\ud t$ is sufficiently small, that approximation remains valid to within the specified tolerance $\varepsilon$ for all later times. Of course, this threshold time increases as $\varepsilon$ decreases, which accounts for the $\varepsilon$-dependence appearing in all of the complexity results. We note that this $T$-independence effect does not preclude an intermediate regime before saturation in which the minimal number of poles still grows with $T$; in practice this growth can resemble $\log T$, as seen in \cref{fig:n_vs_t}. Thus the asymptotic $T$-independent complexity for mild singularities is consistent with the pre-asymptotic behavior often observed in computations, including \cref{fig:n_vs_t}.

\begin{figure}[ht]
    \includegraphics[width=3.25in]{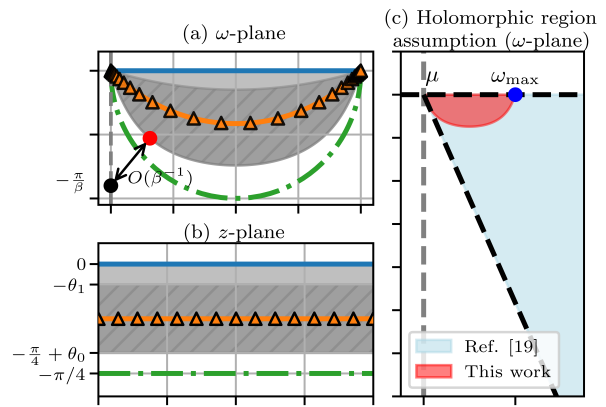}
    \caption{(a,b) Contour transformation from the $\omega$-plane to the $z$-plane. The gray area is the holomorphic region $\Omega$ of $J_{\text{eff}}(\omega)$, and the hatched area is the analyticity strip used in our analysis. Corresponding contours are illustrated in the same color. (c) Comparison of the holomorphic region in our analysis and that in Ref. \cite{ThoennissVilkoviskiyAbanin2024}.
    }
    \label{fig:contour}
\end{figure}

\begin{figure*}[ht]
    \includegraphics[width=170mm]{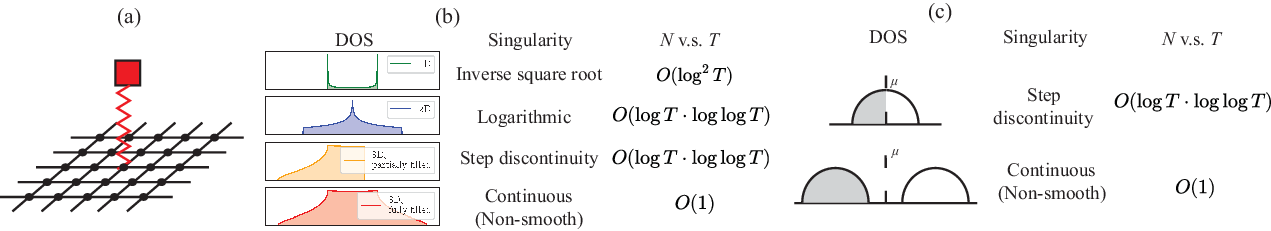}
    \caption{(a) A system coupled to a fermionic periodic lattice. (b) Spectral density in the case of 1D, 2D and 3D lattices exhibiting the van Hove singularities, and corresponding complexity scalings in $T$. (c) Gapless and gapped spectral densities with power-law singularities at the band edge, and corresponding complexity scalings in $T$.}
        \label{fig:VanHove}
\end{figure*}

\emph{Temperature independence.}
The temperature dependence of the BCFs enters through Bose-Einstein or Fermi-Dirac functions in the effective spectral density, affecting the prefactor of exponential convergence. The Fermi-Dirac function is bounded in the analyticity region, so the prefactor (and the final complexity) is $\beta$-independent for fermionic systems. For bosonic systems, the Bose-Einstein function diverges at $\omega=0$, leading after nondimensionalization to a prefactor that scales as $O(1+1/(\beta\omega_c))$ and hence only logarithmic dependence on temperature in the final complexity. In the physically relevant regime $\beta\omega_c\gtrsim 1$, this prefactor remains $O(1)$. For details, see the SM \cite[\cref{sec:looseend}]{supp}.

\emph{Sub-Ohmic, Ohmic and Super-Ohmic bosonic baths.}
The first physical example is the Ohmic family of spectral densities, which has been widely used to model light-matter interactions and quantum dissipation for decades \cite{LeggettChakravartyDorseyetal1987}. The spectral density is given by
\begin{equation}
    J(\omega) = \omega^{\gamma}\omega_c^{-(\gamma+1)}\mathrm e^{-\omega /\omega_c }, \quad \omega \in [0,\infty),\quad \gamma>0.
    \label{eq:Ohmic}
\end{equation}
For $\gamma>1$, $J$ is said to be super-Ohmic, $\gamma=1$ is Ohmic, and $\gamma<1$ is sub-Ohmic. Consider a continuous cutoff function $\chi_W(\omega)$ that is compactly supported on $[0,W]$:
$$
\chi_W(\omega) = \left\{\begin{aligned}
    &\frac{1 -\mathrm e^{-(W - \omega)/ \omega_c } }{1 - \mathrm e^{-W/\omega_c}}, \quad \omega\in[0,W],\\
    &0, \quad \omega \in [W,\infty).
\end{aligned}\right.
$$
We choose a large enough $W$ such that $J(\omega)$ is very close to $J_{\text{trunc}}(\omega) =J(\omega)\chi_W(\omega)$. This finite truncation argument could be made rigorous by carefully bounding the $L^1$ truncation error (see SM \cite[\cref{proposition:truncation}]{supp}).

At zero temperature, we can see that since $\gamma>0$, $J_{\text{eff}}(\omega) = J(\omega)\chi_W(\omega)$ has singularity order $\gamma$ at $\omega=0$ and singularity order $1$ at $\omega=W$. Thus the overall singularity order is $\alpha = \min\{\gamma,1\}$, and  the complexity is independent of $T$ for all Ohmic families. At finite temperature, however, there are qualitative differences for different regimes of $\gamma$. Since $J_{\text{eff}}(\omega)\sim |\omega|^{\gamma-1}/\beta$ as $\omega\to 0$, the singularity order becomes $\alpha = \min\{\gamma-1,1\}$. Thus the super-Ohmic, Ohmic and sub-Ohmic cases each exhibit qualitatively different complexity scalings in $T$ for $L^1$ accuracy: no $T$-dependence for super-Ohmic baths, $N\sim O(\log (T/\varepsilon)\log(\log(T)/\varepsilon))$ for Ohmic baths, and $N\sim O((\log (T/\varepsilon))^2)$ for sub-Ohmic baths.

To the best of our knowledge, this is the first rigorous complexity bound that applies across the Ohmic family while distinguishing the super-Ohmic, Ohmic, and sub-Ohmic regimes through their different $T$-dependences. Our result, as summarized in \cref{tab:summary}, thus provides a novel theoretical perspective on how these regimes differ: the $T$-dependence of the bath-representation complexity is qualitatively distinct in the sub-Ohmic, Ohmic, and super-Ohmic cases.

\emph{Van Hove singularities.}
Next, let us consider a system coupled to an $n$-dimensional periodic lattice, which is a tight-binding lattice of noninteracting electrons, as illustrated in \cref{fig:VanHove}(a). The spectral density exhibits the dimension-dependent van Hove singularities considered above (\cref{fig:VanHove}(b)), all of which are covered by our analysis. In 1D, the spectral density has inverse square-root singularities at the band edges, namely $J_{\text{eff}}(\omega) \sim |\omega-\omega_0|^{-1/2}$, producing the worst-case scaling $N=O\big(\log^2(T/\varepsilon)\big)$. In 2D, the logarithmic singularity $J_{\text{eff}}(\omega) \sim \log|\omega|$ and the step discontinuity at the band edge both give the intermediate scaling $N=O\big(\log(T/\varepsilon)\,\log\log(T/\varepsilon)\big)$. In 3D, if the band is partially filled, the jump discontinuity at the chemical potential $\mu$ leads to the $O(\log(T/\varepsilon)\,\log\log(T/\varepsilon))$ scaling. Otherwise, if the band is fully filled, then the complexity is $T$-independent.

\emph{Gapless and gapped baths.} As a final example, let us illustrate the qualitative difference between gapless and gapped spectral densities for long-time simulation. As shown in \cref{fig:VanHove}(c), the gapped spectral density vanishes continuously at $\omega=\mu$ and thus has singularity order $\alpha>0$. For a gapless bath at finite temperature, the thermal effective spectral density is itself smooth at $\omega=\mu$. However, to obtain a representation whose complexity is uniform in $\beta$, we split the frequency domain at $\mu$ and treat the two sides separately, as described above. In this segmented representation the boundary point $\omega=\mu$ behaves as an endpoint singularity of order $\alpha=0$. By \cref{tab:summary}, reproducing the BCFs for a gapped bath requires a number of modes independent of $T$, while a gapless bath demands bath modes with size growing polylogarithmically in $T$. This polylogarithmic $T$-dependence should therefore be understood as the cost of enforcing $\beta$-independent complexity for a gapless bath, rather than as a physical nonanalyticity of the finite-temperature spectrum itself. For fixed $\beta$, one may instead expect to recover $T$-independent complexity by analyzing the smooth thermal effective spectral density directly, but then the resulting constants would no longer remain uniform as $\beta$ decreases. 

\emph{Applicability to classical open systems.} Although our primary motivation is to characterize the complexity of open quantum systems, the challenge of finding compact bath representations is universal. In classical molecular dynamics, the Generalized Langevin Equation (GLE) governs the system dynamics, where the memory kernel is intrinsically determined by environmental noise fluctuations via the fluctuation-dissipation theorem. Crucially, when this memory kernel is accurately represented by a compact sum of exponentials, the GLEs, which are non-local in time, map exactly onto a system of local auxiliary differential equations~\cite{siegle2010markovian,goychuk2012viscoelastic,kupferman2004fractional}. In particular, the convolution integral with an exponential kernel can be recast as a simple differential equation, effectively eliminating the need to store the history.  This transformation reduces the simulation complexity from quadratic in time $T$ to linear. Such embeddings have proven essential for capturing complex memory effects in viscoelastic phenomena, protein folding, and hydration dynamics~\cite{dalton2023fast,jung2017iterative,daldrop2017external}. We refer the reader to the SM~\cite[\cref{sec:classical}]{supp} for specific examples where spectral singularities naturally arise in these classical contexts.

\emph{Conclusions and outlook.}
In this Letter, we have developed a general theoretical framework for analyzing the complexity of simulating Gaussian environments with general spectral densities. Our result systematically characterizes the complexity for different singularity types. This establishes a rigorous foundation for both recently developed quasi-Lindblad pseudomode methods \cite{HuangLinParkZhu2024,ThoennissVilkoviskiyAbanin2024} and the free-pole HEOM methods \cite{XuYanShietal2022}, as well as Markovian embedding approaches in classical molecular dynamics with memory kernels, precisely quantifying their complexity for long-time simulations at low temperatures.

This work also opens up several interesting directions for future research. Based on our results, we may extend prior analysis of tensor network influence functionals \cite{vilkoviskiy2024bound} beyond the Ohmic case to super-Ohmic and sub-Ohmic spectral densities. Though our analysis could be straightforwardly extended to multi-band systems, it remains unclear what the optimal complexity is in terms of the number of bands. Moreover, it is also unclear how to develop a similar complexity result for the more recent coupled-Lindblad pseudomode theories \cite{HuangParkChanLin2026}, which satisfy additional constraints to guarantee physicality of the system dynamics.

\emph{Acknowledgements.} This material is based upon work supported by the Simons Targeted Grant in Mathematics and Physical Sciences on Moir\'e Materials Magic (Z.H., L.L.), by the U.S. Department of Energy, Office of Science, Accelerated Research in Quantum Computing Centers, Quantum Utility through Advanced Computational Quantum Algorithms, grant no. DE-SC0025572 (L.L.), by a startup grant from the University of Michigan (Z.D., K.W.), by the Van Loo Postdoctoral Fellowship (K.W.), and by the NSF Grant CCF-2312456 (X.L.).  L.L. is a Simons Investigator in Mathematics. The Flatiron Institute is a division of the Simons Foundation.  We thank Andrew Millis, Olivier Parcollet, Steven White, and Maciej Zworski for helpful discussions.

\bibliographystyle{apsrev4-2}
\bibliography{apssamp,gle}
\newpage
\clearpage
\thispagestyle{empty}
\onecolumngrid
\begin{center}
\textbf{\large Support Information for  \\
Long-time simulation of general non-Markovian Gaussian environments at arbitrary temperature: theoretical analysis
}
\end{center}

\begin{CJK*}{UTF8}{mj}
\begin{center}
Zhen Huang$^1$, Zhiyan Ding$^{2}$, Ke Wang$^{2}$, Jason Kaye$^{3,4}$, Xiantao Li$^{5}$ and Lin Lin$^{1,6}$\\
\smallskip
\small{\emph{$^1$\DeptMath\\$^2$\Michigan\\$^3$\FlatironCCQ\\ $^4$\FlatironCCM\\$^5$\Pennstate\\$^6$\LBLMath}}\\
(Dated: \today)

\end{center}
\end{CJK*}
\setcounter{equation}{0}
\setcounter{figure}{0}
\setcounter{table}{0}
\setcounter{page}{1}
\makeatletter
\renewcommand{\theequation}{S\arabic{equation}}
\renewcommand{\thefigure}{S\arabic{figure}}
\renewcommand{\thetable}{S\arabic{table}}
\renewcommand{\bibnumfmt}[1]{[S#1]}
\newtheorem{proposition}{Proposition}

\setcounter{theorem}{0}
\setcounter{lemma}{0}

\setcounter{definition}{0}
\setcounter{assumption}{0}


\onecolumngrid
\setcounter{secnumdepth}{3}
\setcounter{section}{0}
\renewcommand{\thesection}{S\arabic{section}}
\renewcommand{\thesubsection}{S\arabic{section}.\arabic{subsection}}
\renewcommand{\thesubsubsection}{S\arabic{section}.\arabic{subsection}.\arabic{subsubsection}}

\section{Setup}
\label{sec:setup}
In this section, we review the problem setup and perform the nondimensionalization. Recall that in the main text, the BCF is the Fourier transform of the effective spectral density $J_{\text{eff}}(\omega)$:
$$
\Delta(t) = \int_{-\infty}^{\infty} J_{\text{eff}}(\omega) e^{-i\omega t} \mathrm{d}\omega,
$$
where $J_{\text{eff}}(\omega)$ is given by the spectral density $J(\omega)$ through the following relation:
$$
J_{\text{eff}}(\omega) =  \mathrm{sign}(\omega) \frac{J(|\omega|)}{1 - e^{-\beta  \omega } }. \quad \text{For bosons.}
$$
$$
 J^<_{\text{eff}}(\omega) = J(\omega) f_{\mathrm{FD}}(\omega-\mu), \quad J^>_{\text{eff}}(\omega) = J(\omega)  (1 - f_{\mathrm{FD}}(\omega-\mu)), \quad \text{For fermions.} 
$$
Here $J(\omega)$ is the spectral density, which is compactly supported or decays exponentially fast.
As described in the main text, we treat each analytic segment of $J_{\text{eff}}(\omega)$ separately. For fermions, we also treat the domains $\omega>\mu$ and $\omega<\mu$ separately, where $\mu$ is the chemical potential ($\mu$ is zero for bosons). Thus we assume that $J_{\text{eff}}(\omega)$ is supported on $[\omega_a, \omega_b]$ and is analytic in the semi-ellipse below $[\omega_a, \omega_b]$, as illustrated in Fig.~\ref{fig:contour}(a) (the dark grey area). For now, we assume that $\omega_a$ and $\omega_b$ are finite, and we discuss how to truncate infinite domains in \cref{sec:looseend}.
\subsection*{Nondimensionalization}
To nondimensionalize, define $\omega'$, the rescaled frequency that maps $[\omega_a, \omega_b]$ to $[-1, 1]$:
$$
\omega' = \frac{2\omega - (\omega_a + \omega_b)}{\omega_b - \omega_a},\quad \omega\in [\omega_a, \omega_b], \quad \omega'\in[-1,1].
$$
Define the rescaled spectral densities 
$$\widetilde J(\omega') = \frac{\omega_b-\omega_a }{2} J(\omega),\quad \widetilde J_{\text{eff}} (\omega') =\frac{\omega_b-\omega_a }{2} J_{\text{eff}}(\omega).$$
Then the Fourier transform of $\widetilde J_{\text{eff}}(\omega')$, denoted by $\widetilde \Delta(t)$, satisfies $$\widetilde{\Delta}\left(\left(\frac{\omega_b-\omega_a}{2} \right) t \right) = \mathrm e^{\mathrm i\frac{\omega_a + \omega_{b}}{2}t} \Delta\left(t\right).$$ To achieve $L^1$ error less than $\varepsilon$ for $\Delta(t)$ on $[0, T]$, it suffices to achieve $L^1$ error less than $W\varepsilon$ for $\widetilde \Delta(t)$ on $[0, WT]$, where $W = \frac{\omega_b-\omega_a}{2}$ is the half bandwidth. Note that since $\varepsilon$ is the $L^1$ error for $\Delta(t)$, it has the dimension of time. Thus $W\varepsilon$ and $WT$ are dimensionless quantities. In the following, we work with the rescaled quantities $\omega'\in[-1,1]$ and drop the tilde and prime for simplicity.

\subsection*{Contour deformation, change of variables and analyticity assumptions}
A key step in our analysis is the change of variables $\omega = \tanh x$, which maps the domain $\omega\in[-1,1]$ to $x\in[-\infty, \infty]$. Then we have:
$$
  \Delta(t) = \int_{-\infty}^{+\infty}  {J}_{\text{eff}}(\tanh x) \mathrm e^{-\mathrm i t \tanh x} \frac{1}{\cosh^2 x} \ud x.
$$
Let us define
\begin{equation}
    g(z;t) =  J_{\text{eff}} (\tanh z) \mathrm e^{-\mathrm i t \tanh z} \frac{1}{\cosh^2 z},
    \label{eq:g_definition}
\end{equation}
as the complexification of the integrand, for $\operatorname{Im}(z)\in[-\frac{\pi}{4}+\theta_0,0]$. Here $\theta_0>0$.
Then, by contour deformation, we can evaluate $\Delta(t)$ by integrating along the line $\operatorname{Im} z = -y$ for any $y\in [0,\frac{\pi}{4} - \theta_0]$:
\begin{equation}
  \Delta(t) = \int_{-\infty}^{+\infty} g(x - \mathrm i y; t) \ud x,\quad \forall y\in [0, \frac{\pi}{4} - \theta_0].
  \label{eq:contour_deformation}
\end{equation}
As shown in \cref{sec:mainproof}, we will consider $y = -\operatorname{Im}(z) \in [\theta_1, \frac{\pi}{4} - \theta_0]$ for some nonzero $\theta_1$ to take advantage of the decay properties in $t$.  In other words, $g(z;t)$ is required to be analytic in the strip $\mathcal S$:
\begin{equation}
    \mathcal S = \{z:\operatorname{Im} z\in [-\frac{\pi}{4}+\theta_0, -\theta_1]\}.
\end{equation}
Thus we make the following analyticity assumption on $J_{\text{eff}}(\omega)$:
\begin{assumption}
$J_{\text{eff}}(\omega)$ is analytic in the following semi-ellipse $\Omega$ below $[-1,1]$:
$$
\Omega = \{\omega: \omega = \tanh z,  \operatorname{Im} z\in [-\frac{\pi}{4}+\theta_0,0]\}.
$$
Here $\theta_0 \in (0, \frac{\pi}{4})$ is a constant.
\label{assumption:analyticity}
\end{assumption}
The region $\Omega$ is illustrated in Fig.~\ref{fig:contour}(a) (the dark grey area). Since we can always restrict to a smaller holomorphic region, without loss of generality, we may assume that $J_{\mathrm{eff}}(\omega)$ is analytic on $\overline{\Omega}$ except at $\omega=\pm 1$.

\subsection*{Singularity order $\alpha$}
In the main text we have formally introduced the notion of singularity order $\alpha$ to characterize the singularity behavior of $J_{\text{eff}}$ at its endpoint $\omega=\pm 1$. Here we give a more precise definition.

\begin{definition}[Singularity order $\alpha$]   For $\alpha>-1$ and $\alpha\neq 0$, we say a function $f(\omega)$ has a singularity of order $\alpha$ at $\omega = 1$ if there exists $C>0$ such that $|f(\omega)|\le C|\omega-1|^{\alpha}$ as $\omega\to 1$ within $\Omega$, while for every $\alpha'>\alpha$, $\frac{|f(\omega)|}{|\omega-1|^{\alpha'}}$ is unbounded. For $\alpha=0$, we require instead that there exists $C>0$ such that $|f(\omega)|\le C(1+|\log|\omega-1||)$ as $\omega\to 1$ within $\Omega$, and that for every $\alpha'>0$, $\frac{|f(\omega)|}{|\omega-1|^{\alpha'}}$ is unbounded. This class includes both logarithmic singularities and bounded jump discontinuities.  
A similar definition applies to the singularity at $\omega=-1$. For $f (\omega)$ with singularity order $\alpha_{\pm}$ at $\omega=\pm 1$, we define $\alpha = \min\{\alpha_+, \alpha_-\}$ as the overall singularity order.
\label{defn:singularity_type}
\end{definition}
With this definition, after possibly enlarging the constant to account for the compact part of $\Omega$ away from the endpoints, there exists $C_J>0$ such that the corresponding endpoint bounds hold on the two half-domains. Let 
$$\Omega_+ = \{\omega\in \Omega: \operatorname{Re}\omega \ge 0\},\qquad \Omega_- = \{\omega\in \Omega: \operatorname{Re}\omega \le 0\}. $$
 More precisely, $J_{\text{eff}}(\omega)$ of singularity order $\alpha$ satisfies the bounds
\begin{equation}
    \begin{aligned}
         |J_{\text{eff}}(\omega)| \le C_{ J} |\omega-1|^{\alpha}, \quad \text{for } \omega\in \Omega_+,\quad  \alpha>-1, \alpha\neq 0,\\
        |J_{\text{eff}}(\omega)| \le C_{ J} |\omega+1|^{\alpha}, \quad \text{for } \omega\in \Omega_-,\quad  \alpha>-1, \alpha\neq 0,\\
        |J_{\text{eff}}(\omega)| \le C_{ J} (1+ |\log(|\omega-1|)|),  \quad \text{for } \omega\in \Omega_{+}, \quad \alpha=0,\\
        |J_{\text{eff}}(\omega)| \le C_{ J} (1+ |\log(|\omega+1|)|),  \quad \text{for } \omega\in \Omega_{-}, \quad \alpha=0.
    \end{aligned}
\label{eq:J_bound}
\end{equation}

\section{Main result}
\label{sec:mainproof}
Our main result is based on the following theorem, which gives a pointwise estimate on the error of approximating $\Delta(t)$ by a sum of exponentials.
\begin{theorem}[Pointwise error estimate]
    Suppose that the effective spectral function $J_{\text{eff}}(\omega)$ has singularity of order $\alpha>-1$ at $\omega=\pm 1$ (see \cref{defn:singularity_type}). Then there exist constants $A_\alpha,c>0$, independent of $t$, and a choice of $y_2\in (\theta_1,\frac{\pi}{4}-\theta_0)$ such that for any $0<h<1$, $M>1$, and $t\geq 0$, we have
    $$
\left| \Delta(t) - \sum_{|nh|\leq M} h g(nh-\mathrm i y_2;t)\right| \leq C_{ J}\cdot \left\{
\begin{array}{ll}
    A_\alpha(\mathrm e^{-c/h} (1+t)^{-(1+\alpha)} + j_{M,\alpha}(t)), & \alpha\neq 0,\alpha>-1\\
    A_0\left(\mathrm e^{-c/h} \frac{\log(2+t)}{1+t}  + j_{M,0}(t)\right), & \alpha=0.
\end{array}
\right.
    $$
    Here $j_{M,\alpha}(t) = \min\{ \mathrm e^{-(\alpha+1)M}, (1+t)^{-(1+\alpha)}\}$ for $\alpha\neq 0$, $j_{M,0}(t) =  \mathrm e^{- M}$, and $C_{J}$ is the model-dependent constant in \cref{eq:J_bound}.
    \label{thm:pointwise_error}
\end{theorem}
As shown in \cref{thm:pointwise_error}, the approximation of $\Delta(t) = \int_{-\infty}^{\infty} g(x - \mathrm i y; t) \ud x$ is given by a uniform discretization in $x$ with step size $h$ and truncation $|x|<M$, for a fixed $y$. We leave the proof of \cref{thm:pointwise_error} to the end of this section. Note that the tail bound for $\alpha=0$ can be improved to $j_{M,0}(t) = \min (\mathrm e^{- M} , \frac{\log^2(2+t)}{1+t})$, but using $j_{M,0}(t) =  \mathrm e^{-M}$ is sufficient for this work. The pointwise error estimate in \cref{thm:pointwise_error} directly leads to the main result of this work, which is the $L^1$ error estimate for approximating $\Delta(t)$:
\begin{corollary}[Main result: $L^1$ error estimate]
  With the same assumptions as in \cref{thm:pointwise_error}, to achieve $L^1$ error less than $\varepsilon$ for approximating $\Delta(t)$ by a sum of exponentials, the number of terms required satisfies the following scaling:
  $$
  \begin{array}{ll}
    N = O\left(\log^2 (C_J /  \varepsilon)\right), & \text{for } \alpha> 0,\\
    N = O\left(\log\left(C_J \frac{\log( T)}{ \varepsilon}\right) \log\left(\frac{C_J T}{\varepsilon}\right)\right), & \text{for } \alpha=0,\\
    N = O\left(\log^2(C_J T/\varepsilon)\right), & \text{for } -1<\alpha<0.
  \end{array}
  $$
  \label{cor:main}

\end{corollary}

  In unrescaled quantities, $T$ and $\varepsilon$ should be replaced by $WT$ and $W\varepsilon$ respectively. Thus the scaling becomes (omitting the $C_J$ factor):
\begin{equation}
    \begin{array}{ll}
    N = O\left(\log^2 (1/ W  \varepsilon)\right), & \text{for } \alpha> 0,\\
    N = O\left(\log\left(\frac{\log(W T)}{ W\varepsilon}\right) \log\left(\frac{T}{\varepsilon}\right)\right), & \text{for } \alpha=0,\\
    N = O\left(\log^2(  T/\varepsilon)\right), & \text{for } -1<\alpha<0.
    \end{array}
    \label{eq:main}
\end{equation}

Corollary \ref{cor:main} leads to the main conclusion of this work in \cref{tab:summary}, for which we have rescaled the quantities back to the original scale, and explicitly considered $C_J$'s dependence on $\beta$, the latter of which is addressed in \cref{sec:looseend}. We have also replaced $W$ with the characteristic frequency $\omega_c$, which is straightforward for spectral densities with finite support, and for spectral densities with infinite support but exponential decay, this is also addressed in \cref{sec:looseend}.

The proof of \cref{cor:main} from \cref{thm:pointwise_error} is simply to integrate the pointwise error estimate over $t\in[0,T]$, and choose $M$ and $h$ to bound both terms in the pointwise error estimate by $\varepsilon$. The details are given below.
\begin{proof}[Proof of Corollary \ref{cor:main} from Theorem \ref{thm:pointwise_error}]
 For $\alpha>0$, by decomposing the integral into $[0, \mathrm e^M]$ and $[\mathrm e^M, T]$, and using the two bounds $\mathrm e^{-(\alpha+1)M}$ and $(1+t)^{-(1+\alpha)}$ correspondingly for the two intervals, we have that $ \int_0^T |j_{M,\alpha}(t)| \ud t\leq \frac{\alpha+1}{\alpha} \mathrm e^{- \alpha  M }$.
    Thus using Theorem \ref{thm:pointwise_error}, the $L^1$ error is bounded by $C_J\left(\frac{A_\alpha}{\alpha} \mathrm e^{-c / h} + \frac{A_\alpha(\alpha+1)}{\alpha} \mathrm e^{-\alpha M}\right)$.  By choosing $M =  \frac{1}{\alpha} \log\left(\frac{2(\alpha+1)C_JA_\alpha}{\alpha\varepsilon}\right)  $, $1/h = \frac{1}{c} \log\left(\frac{2C_{ J}A_\alpha}{\alpha  \varepsilon}\right)$, the $L^1$ error is less than $\varepsilon$, and the number of terms $N =\lceil 2M/h\rceil = O(\log^2(C_J/\varepsilon))$. 
    
    For $\alpha<0$, since $\int_0^T (1+t)^{-(1+\alpha)} \ud t =  ((1+T)^{|\alpha|} - 1)/|\alpha|$, we have that the $L^1$ error is bounded by
$$C_{J}\left(\frac{A_\alpha}{|\alpha|} \mathrm e^{-c/h}(1+ T)^{|\alpha|} +A_\alpha  T \mathrm e^{-(\alpha+1)M}\right).$$ To make this bounded by $\varepsilon$, we can choose $ M = \frac{1}{\alpha+1} \log\left(\frac{2C_{  J}A_\alpha T}{ \varepsilon}\right)$, $ 1/h = \frac{1}{c} \log\left(\frac{2C_{ J}A_\alpha(1+ T)^{|\alpha|}}{|\alpha| \varepsilon}\right)$. Thus $N = \lceil 2M/h\rceil = O(\log^2(C_J T/\varepsilon))$.  

Finally, for $\alpha=0$, since $\int_0^T \frac{\log(2+t)}{1+t} \leq 2\int_0^T \frac{\log(2+t)}{2+t} \ud t \leq  (\log(2+T))^2$, the $L^1$ error is bounded by $$C_J\left(\frac{A_0}{2} \mathrm e^{-c/h} (\log(2+T))^2 + A_0 T \mathrm e^{-M}\right).$$ Thus, it suffices to choose $M = \log\left(\frac{2C_J A_0 T}{\varepsilon}\right)$, $1/h = \frac{1}{c} \log\left(\frac{2C_J A_0 (\log(2+T))^2}{\varepsilon}\right)$. We conclude that $$N = \lceil 2M/h\rceil = O\left(\log\left(C_J \frac{\log( T)}{ \varepsilon}\right) \log\left(\frac{C_J T}{\varepsilon}\right)\right).$$
\end{proof}

Another consequence of \cref{thm:pointwise_error} is the following $L^\infty$ error estimate, which is $T$-independent for all $\alpha>-1$:
\begin{corollary}[$L^\infty$ error estimate]
    With the same assumptions as in \cref{thm:pointwise_error}, to achieve $L^\infty$ error less than $\varepsilon_\infty$ for approximating $\Delta(t)$ by a sum of exponentials, the number of terms required satisfies the following scaling: $N = O\left(\log^2 (C_J/\varepsilon_\infty)\right)$ for all $\alpha>-1$.
    \label{cor:linfty}
\end{corollary}
\begin{proof}[Proof of Corollary \ref{cor:linfty} from Theorem \ref{thm:pointwise_error}]
By Theorem \ref{thm:pointwise_error}, the $L^\infty$ error is bounded by $C_J\cdot A_\alpha(\mathrm e^{-c/h} + \mathrm e^{-(\alpha+1)M})$ for $\alpha\neq 0$ and by $C_J\cdot A_0(\mathrm e^{-c/h} + \mathrm e^{-M})$ for $\alpha=0$. Thus by choosing $M = \frac{1}{\alpha+1} \log\left(\frac{2C_J A_\alpha}{\varepsilon_\infty}\right)$ and $1/h = \frac{1}{c} \log\left(\frac{2C_J A_\alpha}{\varepsilon_\infty}\right)$ when $\alpha\neq 0$, and $M = \log\left(\frac{2C_J A_0}{\varepsilon_\infty}\right)$ and $1/h = \frac{1}{c} \log\left(\frac{2C_J A_0}{\varepsilon_\infty}\right)$ when $\alpha=0$, the $L^\infty$ error is less than $\varepsilon_\infty$, and in both cases $N = \lceil 2M/h\rceil = O(\log^2(C_J/\varepsilon_\infty))$.
\end{proof}

Now let us give a proof of \cref{thm:pointwise_error}. The error estimate in \cref{thm:pointwise_error} has two parts: the discretization error on an infinite grid of size $h$ and the truncation error on $[-M,M]$. Both errors are controlled by the decay properties of $g(z;t)$ in the strip $\mathcal S$, as stated in the following lemma:
\begin{lemma}[Upper bound of $g(x-\mathrm iy;t)$ (\cref{eq:g_definition})]  Given a singularity order $\alpha>-1$ of $ J_{\text{eff}}(\omega )$, the following upper bound holds for any $t\geq 0$ and $y\in [\theta_1, \frac{\pi}{4} - \theta_0]$:
    $$
|g(x-\mathrm iy;t)| \le C_{  J} \left\{
\begin{array}{ll}
    C_\alpha p_\alpha(x;t) , & \alpha>-1, \alpha\neq 0, \\
    C_0 (1+|x|)p_0(x;t)  , & \alpha=0.
\end{array}
 \right.
$$
where $p_\alpha(x;t)$ (for $\alpha>-1$)   are defined as 
\begin{equation}
    p_\alpha(x;t) = \mathrm e^{-c_{\theta_1} \frac{t}{\cosh^2 x}} \cosh^{-2(\alpha+1)} x,
    \label{eq:h_defn}
\end{equation}
and $C_\alpha$ is a constant that only depends on $\alpha$,  $C_{  J}$  is the constant in \cref{eq:J_bound}, and $c_{\theta_1} = \sin(2\theta_1)/2>0$.
\label{lem:g}
\end{lemma}
\begin{proof}[Proof of \cref{lem:g}] Let $z = x - \mathrm iy$.
   Recall that $g(z;t) =   J_{\text{eff}} (\tanh z) \mathrm e^{-\mathrm i t \tanh z} \frac{1}{\cosh^2 z}$. Since $|y|<\frac{\pi}4$, $
    |\cosh z|^2 = \cosh^2 x \cos^2 y + \sinh^2 x \sin^2 y \geq \frac 1 2 \cosh^2 x
$. Therefore $|1/\cosh^2 z| \leq 2\cosh^{-2} x$. Next, we bound $\mathrm e^{-\mathrm i t \tanh z} $. Note that $|\mathrm e^{-\mathrm it\tanh z}| = \mathrm e^{t\operatorname{Im}(\tanh z)}$, and $\operatorname{Im}(\tanh (x-\mathrm iy)) = \frac{-\sin 2y}{2(\sinh^2 x+ \cos^2y)}$.  Using that $y\in[\theta_1, \frac{\pi}{4} - \theta_0]$, we have $-\sin 2y\leq -\sin 2\theta_1 < 0$, and $\sinh^2x +  \cos^2y\leq\sinh^2x+1 = \cosh^2x$. Thus $|\mathrm e^{-\mathrm it\tanh z}| \leq \mathrm e^{-c_{\theta_1} t/\cosh^2 x}$, where $c_{\theta_1} = \sin(2\theta_1)/2>0$.  To pass from the endpoint bounds in \cref{eq:J_bound} to a bound on $J_{\text{eff}}(\tanh z)$, note that
$$
\operatorname{Re}(\tanh (x-\mathrm iy)) = \frac{\sinh(2x)}{\cosh(2x)+\cos(2y)},
$$
so the sign of $\operatorname{Re}(\tanh z)$ is the sign of $x$. Hence $\omega=\tanh z$ lies in $\Omega_+$ for $x\ge 0$ and in $\Omega_-$ for $x\le 0$. Moreover,
$$
1-\tanh z = \frac{\mathrm e^{-z}}{\cosh z}, \qquad 1+\tanh z = \frac{\mathrm e^{z}}{\cosh z}.
$$
Since $2^{-1/2}\cosh x\le |\cosh z|\le \cosh x$ when $|y|<\pi/4$, for $x\ge0$ we obtain
$$
\frac{2}{1+\mathrm e^{2x}} \le |\omega-1| \le \frac{2\sqrt{2}}{1+\mathrm e^{2x}},
$$
and similarly for $x\le0$,
$$
\frac{2}{1+\mathrm e^{-2x}} \le |\omega+1| \le \frac{2\sqrt{2}}{1+\mathrm e^{-2x}}.
$$
These quantities are comparable to $\cosh^{-2}x$ up to absolute constants, and therefore \cref{eq:J_bound} yields the following bound for $J_{\text{eff}}(\tanh z)$: 
\begin{equation}
    |  J_{\text{eff}} (\tanh z)| \leq  {C_{  J}} \cdot \left\{
    \begin{array}{ll}
        C_\alpha \cosh^{-2\alpha} x, & \alpha\neq 0, \alpha>-1,\\
        C_0 (1 + |x|), & \alpha=0.
    \end{array}
    \right.
\end{equation}
Here $C_\alpha$ is a constant that only depends on $\alpha$. Combining the bounds for the three factors in $g(z;t)$, we have the desired bound for $|g(z;t)|$.
\end{proof}
Finally, we prove \cref{thm:pointwise_error} using the upper bound of $g(z;t)$ in \cref{lem:g}.
\begin{proof}[Proof of \cref{thm:pointwise_error}]
  Let us first show the following discretization error estimate for the infinite grid:
    $$
\left| \Delta(t) - \sum_{n=-\infty}^{\infty} h g(nh-\mathrm i y_2;t)\right| \leq C_{ J}\cdot \left\{
\begin{array}{ll}
    A_\alpha \mathrm e^{-c/h} (1+t)^{-(1+\alpha)}  , & \alpha\neq 0,\alpha>-1,\\
    A_0  \mathrm e^{-c/h} \frac{\log(2+t)}{1+t}  , & \alpha=0.
\end{array}
\right.
    $$
 Choose $$
    y_2 = \frac{\theta_1 + (\frac{\pi}{4}-\theta_0)}{2}, \qquad a = \frac{\frac{\pi}{4}-\theta_0-\theta_1}{2} >0.
    $$
    Then $g(\cdot-\mathrm i y_2;t)$ is analytic in the strip $\{z:|\operatorname{Im} z|<a\}$. Using \cite[Theorem 5.1]{TrefethenWeideman2014}, a uniform grid with step size $h$ has the following error estimate:
  $$
  \left| \Delta(t) - \sum_{n=-\infty}^{\infty} h g(nh-\mathrm i y_2;t)\right| \leq \frac{2M_g}{\mathrm e^{2\pi a/h} - 1},
  $$
 Here $M_g = \sup_{y\in [\theta_1, \frac{\pi}{4} - \theta_0]} \int_{-\infty}^{\infty} |g(x-\mathrm i y;t)| \ud x$.  With \cref{lem:g}, we only need to show the following integral bounds for $p_\alpha(x;t)$ (for $\alpha>-1$, $\alpha\neq 0$)  and $p_0(x;t)$:
  $$
  \int_{-\infty}^{\infty} p_\alpha(x;t) \ud x \leq a_\alpha  (1+t)^{-(1+\alpha)}, \quad \int_{-\infty}^{\infty} (1+|x|)p_0(x;t) \ud x \leq a_0 \frac{\log(2+t)}{1+t}.
  $$
  Here $a_\alpha$ is a constant that only depends on $\alpha$. Conducting change of variables $v = 1 - \tanh^2 x$, we have:
  $$
  \int_{-\infty}^{\infty} p_\alpha(x;t) \ud x = \int_0^1 \mathrm e^{-c_{\theta_1} t v} v^{\alpha} \frac{\mathrm dv}{\sqrt{1-v}}
  $$
    For $v\in[0,\frac 1 2]$, we have 
    $$\int_0^{\frac 1 2} \mathrm e^{-c_{\theta_1} t v} v^\alpha \frac{\ud v}{\sqrt{1-v}} \leq \sqrt{2}\mathrm e^{c_{\theta_1}}\int_0^{\frac 1 2} \mathrm e^{-c_{\theta_1} (1+t) v} v^\alpha \ud v\leq \sqrt{2}\mathrm e^{c_{\theta_1}}\int_0^{+\infty} \mathrm e^{-c_{\theta_1} (1+t) v} v^\alpha \ud v = \frac{\sqrt{2}\mathrm e^{c_{\theta_1}}}{c_{\theta_1}^{\alpha+1}}\frac{\Gamma(\alpha+1)}{(1+t)^{\alpha+1}}.$$ 
    For $v\in[\frac 1 2,1]$, $\mathrm e^{-c_{\theta_1} t v} \leq \mathrm e^{-\frac{c_{\theta_1} t}{2}}$. Thus $$\int_{\frac 1 2}^1 \mathrm e^{-c_{\theta_1} t v} v^\alpha \frac{\ud v}{\sqrt{1-v}} \leq \mathrm e^{-\frac{c_{\theta_1} t}{2}} \int_{\frac 1 2}^1 v^\alpha \frac{\ud v}{\sqrt{1-v}} \leq B_\alpha \mathrm e^{-\frac{c_{\theta_1} t}{2}},$$ where $B_\alpha$ could be taken as $\int_{0}^1 v^\alpha \frac{\ud v}{\sqrt{1-v}} = \frac{\sqrt{\pi}\Gamma(\alpha+1)}{\Gamma(\alpha+\frac 3 2)}$. Combining the two parts, we have that  $\int_{-\infty}^{\infty} p_\alpha(x;t) \ud x$ is bounded by $(1+t)^{-(1+\alpha)}$ up to a constant for all $\alpha>-1$.

 For $\alpha=0$, by symmetry and the same change of variables $v=1-\tanh^2 x$,
    $$
    \int_{-\infty}^{\infty} (1+|x|)p_0(x;t) \,\ud x = \int_0^1 \left(1+\operatorname{arctanh}(\sqrt{1-v})\right) \mathrm e^{-c_{\theta_1}tv}\frac{\ud v}{\sqrt{1-v}}.
    $$
    Since $\operatorname{arctanh}(\sqrt{1-v}) \leq C(1+|\log v|)$ for $v\in(0,1)$, splitting again into $[0,\frac12]$ and $[\frac12,1]$ gives
    $$
    \int_{-\infty}^{\infty} (1+|x|)p_0(x;t) \,\ud x \leq C\int_0^{1/2} \mathrm e^{-c_{\theta_1}tv}(1+|\log v|)\,\ud v + C\mathrm e^{-c_{\theta_1}t/2}.
    $$
    After the rescaling $s=(1+t)v$, the first term is bounded by $C\frac{\log(2+t)}{1+t}$, and the second term is dominated by the same quantity. Hence
    $$
    \int_{-\infty}^{\infty} (1+|x|)p_0(x;t) \,\ud x \leq a_0\frac{\log(2+t)}{1+t}.
    $$

    Next, we only need to bound the truncation error from $|x|>M$:
    $
    \sum_{|nh|>M} h |g(nh-\mathrm i y_2;t)| \leq C_{ J}\cdot A_\alpha j_{M,\alpha}(t)
    $,
    where $j_{M,\alpha}(t) = \min\{ \mathrm e^{-(\alpha+1)M}, (1+t)^{-(1+\alpha)}\}$. (For $\alpha=0$, $j_{M,0}(t) = \mathrm e^{-M}$.) This amounts to the following tail bound for $p_\alpha(x;t)$:
    $$
\sum_{|nh|>M} h p_\alpha(nh;t) \leq A_\alpha j_{M,\alpha}(t),\quad \sum_{|nh|>M} h (1+|nh|)p_0(nh;t) \leq A_0 j_{M,0}(t).
    $$

    Let us first prove the $\alpha\neq 0$ case.
   The $\mathrm e^{-(\alpha+1)M}$ bound is straightforward. Since $\cosh s\geq \frac12\mathrm e^{|s|}$, for $\alpha\neq 0$ we have
    $$
    p_\alpha(nh;t)\leq \cosh^{-2(\alpha+1)}(nh) \leq 2^{2(\alpha+1)}\mathrm e^{-2(\alpha+1)|nh|} \leq 2^{2(\alpha+1)}\mathrm e^{-(\alpha+1)|nh|}.
    $$ 
    Since for any $\alpha>-1$, the following holds:
 
    $$
    h\sum_{|nh|>M} \mathrm e^{-(\alpha+1)|nh|} \leq 2h\sum_{n> M/h} \mathrm e^{-(\alpha+1)nh} \leq \frac{2h}{1-\mathrm e^{- (\alpha+1) h}}\mathrm e^{- (\alpha+1) M}.
    $$
    Since $0<h<1$, convexity gives $1-\mathrm e^{- (\alpha+1) h}\geq h(1-\mathrm e^{- (\alpha+1)})$, so the prefactor is bounded by a constant depending only on $\alpha$. This gives the desired $\mathrm e^{-(\alpha+1)M}$ bound for $\alpha\neq 0$.

    Next, we prove the $(1+t)^{-(1+\alpha)}$ bound for  $\alpha\neq 0$. Note that we have shown that $\int_{-\infty}^\infty p_\alpha(x;t)\ud x \leq A_\alpha (1+t)^{-(1+\alpha)}$.  In other words, it suffices to show that the discrete sum is bounded by the same quantity up to a constant. To this end, choose the analytic branch
$$
p_\alpha(z;t)=\exp\!\left(-c_{\theta_1}\frac{t}{\cosh^2 z}\right)\left(\frac{1}{\cosh^2 z}\right)^{\alpha+1}
$$
on the strip $\{z:|\operatorname{Im} z|<\frac{\pi}{8}\}$. This is well defined because $\cosh z$ has no zeros there, so $\cosh^{-2} z$ is analytic and nonvanishing on a simply connected domain. We then estimate
$$
\left|\int_{-\infty}^\infty p_\alpha(x;t)\ud x- \sum_{k=-\infty}^{\infty} h p_\alpha(kh;t) \right|.
$$
 Note that for such $z$, we have $\frac 1 2 \cosh^2x\leq |\cosh^2z| \leq 2 \cosh^2 x $, thus $p_\alpha(z;t)$ is bounded by 
$|p_\alpha(z;t)| \leq 2^{(1+\alpha)} \mathrm e^{-c_{\theta_1} \frac{t}{2\cosh^2 x}} \cdot \cosh^{-2(\alpha+1)} x $.
($x = \operatorname{Re}(z)$).
Thus $|p_\alpha(z;t)|$ is bounded by $2^{1+\alpha}p_\alpha(x;t/2)$. Since we have proved in the previous lemma that $\int_{-\infty}^\infty |p_\alpha(x;t)| \ud x  $ is bounded by $(1+t)^{-(\alpha+1)}$ up to a constant, we have $\int_{-\infty}^\infty |p_\alpha(x-\mathrm iy;t)| \ud x \leq a_\alpha' (1+t/2)^{-(\alpha+1)}\leq a_\alpha'' (1+t )^{-(\alpha+1)}$ for all $y\in[-\frac{\pi}{8}, \frac{\pi}{8}]$ for some constant $a_\alpha', a_\alpha''$ independent of $t$.  Thus using again \cite[Theorem 5.1]{TrefethenWeideman2014}, for $h<1$,
we have the following bound:
$
\left|\int_{-\infty}^\infty p_\alpha(x;t)\ud x- \sum_{k=-\infty}^{\infty} h p_\alpha(kh;t) \right| \leq \widetilde a_\alpha (1+t )^{-(\alpha+1)}
$,
 where $\widetilde a_\alpha$ is a constant independent of $t$ and $h$; here we only use that $h<1$, so $(\mathrm e^{2\pi a/h}-1)^{-1}$ is bounded by an absolute constant. Therefore
$$
h\sum_{n=-\infty}^\infty p_\alpha(nh;t) \leq \int_{-\infty}^\infty p_\alpha(x;t)\ud x + \widetilde a_\alpha (1+t )^{-(\alpha+1)} \leq \widetilde A_\alpha (1+t)^{-(\alpha+1)},
$$
which in turn controls the tail sum $h\sum_{|nh|\geq M} p_\alpha(nh;t)$.

    Finally, for $\alpha=0$,
    $$
    (1+|nh|)p_0(nh;t)\leq (1+|nh|)\cosh^{-2}(nh) \leq 4(1+|nh|)\mathrm e^{-2|nh|} \leq 4\mathrm e^{-|nh|}.
    $$
    Then, similar to the $\alpha\neq 0$ case, we have
    $
    h\sum_{|nh|>M} \mathrm e^{- |nh|}  \leq \frac{2h}{1-\mathrm e^{-  h}}\mathrm e^{- M}\leq \frac{2}{1-\mathrm e^{-  1}}\mathrm e^{-M}
    $. This gives the desired $\mathrm e^{-M}$ bound for $\alpha=0$.

\end{proof}

\section{Temperature (in)dependence and finite frequency truncation}
\label{sec:looseend}
\subsection*{Fermionic temperature independence}
Let us first address temperature independence for fermionic baths. This addresses the boundedness of the Fermi-Dirac function in our analytic domain. Since we have already separated the frequency domain into $\omega>\mu$ and $\omega<\mu$, the Fermi-Dirac function can appear only at the left or right endpoint of the support of $J_{\text{eff}}(\omega)$, namely:
$$
J_{\text{eff}}(\omega) = J(\omega) f_{FD}(\omega\pm 1), \quad \text{or } J_{\text{eff}}(\omega) = J(\omega) (1 - f_{FD}(\omega \pm 1)).
$$

The temperature-independence result thus stems from the following proposition for the Fermi-Dirac distribution, which states its boundedness in the complex domain $\Omega$ regardless of $\beta$.
\begin{proposition}
   The Fermi-Dirac function $f_{\beta}(\omega) = \frac{1}{\mathrm e^{\beta (\omega+1)}+1}$ is upper bounded by a constant independent of $\beta$ for all $\omega\in \Omega$, where $\Omega$ is the analytic domain defined in Assumption \ref{assumption:analyticity}.
\end{proposition}
\begin{proof}
  The function $f_\beta(\omega)$ satisfies the following upper bound for some constant $C$ independent of $\beta$:
\begin{equation}
  |f_\beta(\omega)|\leq \left\{\begin{array}{cc}
     \frac{C}{\beta d},    &  d\leq \frac 1 {2\beta}\\
      C,   & d\geq \frac 1 {2\beta}
    \end{array}\right.
\end{equation}
Here $d$ is the distance between $\omega$ and the poles in $\{-1+ \mathrm i(2n+1)\pi/\beta, n\in\mathbb{Z}\}$ that are closest to $\omega$.
 Now write $\omega=\tanh(x-\mathrm i y)$ with $y\in[0,\frac{\pi}{4}-\theta_0]$. Using
$$
\omega+1=\frac{e^{2x}+\cos(2y)-\mathrm i\sin(2y)}{\cosh(2x)+\cos(2y)},
$$
we see that $\omega+1$ lies in the sector $\{re^{-\mathrm i\phi}: r\ge 0,\ 0\le \phi\le \frac{\pi}{2}-2\theta_0\}$, because
$$
|\arg(\omega+1)|=\arctan\!\left(\frac{\sin(2y)}{e^{2x}+\cos(2y)}\right)\le 2y\le \frac{\pi}{2}-2\theta_0.
$$
Hence the closest pole to $\omega$ is $-1-\mathrm i\pi/\beta$, and the distance from this pole to the above sector is bounded below by $(\pi/\beta)\sin(2\theta_0)$. Therefore
$$
d \geq \frac{\pi}{\beta}\sin(2\theta_0).
$$
This $O(1/\beta)$ separation between the poles and $\Omega$ is also illustrated in \cref{fig:contour}(a).
Thus $f_\beta(\omega)$ is upper bounded by $C\max\{1, \frac{1}{\pi\sin 2\theta_0}\}$ 
for $\omega\in \Omega$, which is independent of $\beta$.
\end{proof}
Proposition 1 thus implies that the temperature dependence of $C_J$ for fermionic baths is at most a constant factor, which does not affect the scaling of $N$ with respect to $\varepsilon$ and $T$.

\subsection*{Bosonic temperature dependence}
For bosonic problems, recall that $J_{\mathrm{eff}}(\omega)$ is given by
$$
J_{\mathrm{eff}}(\omega) = \mathrm{sign}(\omega)\frac{J(|\omega|)}{1-\mathrm e^{-\beta \omega}}.
$$
Note that for $\omega>0$, $\frac{1}{1-\mathrm e^{-\beta \omega}} = 1+ \frac{\mathrm e^{-\beta \omega}}{1-\mathrm e^{-\beta \omega}}\leq 1 + \frac{1}{\beta \omega}$. For $\omega<0$, $\frac{1}{\mathrm e^{\beta\cdot(- \omega)}-1} \leq \frac{1}{\beta\cdot(-\omega)}$. Thus for a given $\beta$, $J_{\mathrm{eff}}(\omega)$ is given by:
$$
J_{\mathrm{eff}}(\omega) =  \left\{
  \begin{aligned}
    J(\omega)  + J_\beta(\omega), \quad &\text{for } \omega>0,\\
    J_\beta(\omega), \quad &\text{for } \omega<0,
  \end{aligned}
\right.\quad \text{where } |J_\beta(\omega)| \leq \frac{J(|\omega|)}{\beta |\omega|}.
$$
Here we have used the notation $J_\beta(\omega)$ to denote the temperature-dependent part of $J_{\mathrm{eff}}(\omega)$. Since $|J_\beta(\omega)| \leq \frac{J(|\omega|)}{\beta |\omega|}$, we know that this temperature-dependent part $J_\beta(\omega)$ has a singularity order that is one order lower than the singularity order of $J(\omega)$ at $\omega=0$, and has a preconstant that is $O(1/\beta)$ times the preconstant of $J(\omega)$. Hence the temperature-dependent contribution carries an $O(1/\beta)$ prefactor and lowers the local singularity order by one at $\omega=0$. 

As shown in \cref{tab:summary}, for super-Ohmic behavior this gives only logarithmic dependence on $1/\beta$, whereas for Ohmic or sub-Ohmic behavior the reduced singularity order also changes the relevant $T$-dependence.

\subsection*{Finite frequency truncation of Ohmic-family spectral densities}
Here we address how to truncate the infinite support of exponentially decaying spectral densities, such as the Ohmic-family spectral densities. As mentioned in the main text, we introduce the following:
$$
J_{\mathrm{trunc}}(\omega) = J(\omega)\chi_W(\omega),
$$
where $J(\omega) = \omega^{\gamma}\omega_c^{-(\gamma+1)}\mathrm e^{-\omega/\omega_c}$ and $\chi_W(\omega)$ is a cutoff function so that $J_{\mathrm{trunc}}(\omega)$ is compactly supported on $[0,W]$. We choose $\chi_W(\omega)$ as follows:
$$
\chi_W(\omega) = \frac{ (1 - \mathrm e^{-(W-\omega)/\omega_c})}{  1-\mathrm e^{- W/\omega_c}}, \quad \text{for } \omega\in[0,W]; \quad \chi_W(\omega) = 0, \quad \text{for } \omega\in[W,+\infty).
$$
So the error of this truncation is given by:
$$
\Delta_{\text{error}}(t) = \int_0^{\infty} J(\omega) (1-\chi_W(\omega)) e^{-\mathrm i \omega t} \mathrm{d}\omega.
$$
The following proposition establishes that, for this truncation to achieve an $L^1$ error below $\varepsilon$, it suffices to choose $W = O(\log(1/\varepsilon))$. This contributes only $O(\log\log(1/\varepsilon))$ to the scaling of $N$ with respect to $\varepsilon$, since $N$ scales at most logarithmically with $W$ as shown in \cref{tab:summary}.
\begin{proposition}
 There exist constants $C_{\gamma, \omega_c}, W_0>0$ independent of $W$ and $t$ such that for any $W\geq W_0$ and $t\geq 0$, we have $|\Delta_{\text{error}}(t) |\leq C_{\gamma, \omega_c} \mathrm e^{-W/2\omega_c}\frac{1}{(1+t)^2}$. 
    \label{proposition:truncation}
\end{proposition}

\begin{proof}[Proof of Proposition \ref{proposition:truncation}] Let $J_{\mathrm{error}}(\omega) = J(\omega) (1-\chi_W(\omega))$. Since $|J_{\mathrm{error}}(\omega)|\leq | J(\omega)|$ and $J(\omega)$ is in $L^1([0,\infty])$, it follows that $|J_{\mathrm{error}}(\omega)|$ is also in $L^1([0,\infty])$. Therefore $\Delta_{\mathrm{error}}(t)$ is bounded.

To establish the decay of $\Delta_{\mathrm{error}}(t)$, note that $J_{\mathrm{error}}(\omega)$ is continuous and twice differentiable on $[0,W]$ and on $[W,+\infty)$. Using integration by parts twice, we have that for any $t\geq 0$,
$$
\int_0^\infty J_{\mathrm{error}}(\omega) \mathrm e^{-\mathrm i\omega t} \ud \omega  = -\frac{1}{t^2}\int_0^\infty J_{\mathrm{error}}''(\omega) \mathrm e^{-\mathrm i\omega t} \ud \omega -\frac{1}{t^2}\left( J_{\mathrm{error}}'(W^+) - J_{\mathrm{error}}'(W^-)\right)\mathrm e^{-\mathrm i W t}.
$$
Since $J_{\mathrm{error}}(\omega)$ is given by: $$
J_{\mathrm{error}} (\omega) = \omega_c^{-(\gamma+1)} \left\{
\begin{array}{ll}
\frac{\mathrm e^{-W/\omega_c}}{1 - \mathrm e^{-W/\omega_c}} \omega^{\gamma}( 1- \mathrm e^{-\omega/\omega_c}) & \omega \in [0,W],\\
\omega^{\gamma} \mathrm e^{-\omega/\omega_c} & \omega \in [W,+\infty).
\end{array}
\right.
$$
Note that $J_{\mathrm{error}} (\omega)\sim O(\omega^{\gamma + 1}) $ as $\omega \to 0$, so $J_{\mathrm{error}}''(\omega)\sim O(\omega^{\gamma - 1})$ as $\omega \to 0$. For $\gamma>0$, this implies that $J_{\mathrm{error}}''(\omega)$ is integrable near $0$. Moreover, on $[0,W]$ one has
$$
\int_0^W |J_{\mathrm{error}}''(\omega)| \, \mathrm d\omega \leq C_{\gamma,\omega_c}\frac{\mathrm e^{-W/\omega_c}}{1 - \mathrm e^{-W/\omega_c}} (1+W^\gamma),
$$
while on $[W,+\infty)$,
$$
\int_W^{+\infty}|J_{\mathrm{error}}''(\omega)|\,\mathrm d\omega \leq C_{\gamma,\omega_c}' \mathrm e^{-W/\omega_c}(1+W^\gamma).
$$
Also,
$$
|J_{\mathrm{error}}'(W^+) - J_{\mathrm{error}}'(W^-) | = |J(W)| \frac{ 1}{(1-\mathrm e^{-  W/\omega_c})\omega_c}.
$$
Hence both the $L^1$ norm of $J_{\mathrm{error}}''$ and the jump term are bounded by $C_{\gamma,\omega_c}'' \frac{\mathrm e^{-W/\omega_c}}{1-\mathrm e^{-W/\omega_c}}(1+W^\gamma)$. This is not uniform in $W$ near $W=0$ when $0<\gamma<1$, but for sufficiently large $W$ it is bounded by $\widetilde C_{\gamma,\omega_c}\mathrm e^{-W/(2\omega_c)}$. Therefore, for $W\ge W_0$,
$$
|\Delta_{\mathrm{error}}(t)| \leq C_{\gamma, \omega_c} ' \mathrm e^{-W/2\omega_c}\frac{1}{t^2},
$$
and combining this with the boundedness of $\Delta_{\mathrm{error}}(t)$ yields
$$
|\Delta_{\mathrm{error}}(t)| \leq C_{\gamma, \omega_c} \mathrm e^{-W/2\omega_c}\frac{1}{(1+t)^2}, \qquad t\ge 0,
$$
for some constant $C_{\gamma, \omega_c}$ independent of $W$ and $t$. 
\end{proof}

The above proposition establishes that, asymptotically, it is sufficient to choose $W = O(\log(1/\varepsilon))$ for all $\gamma>0$ at zero temperature. For finite temperature, because the Bose--Einstein distribution reduces the singularity order by one, the same cutoff holds for all $\gamma>1$ at any temperature. For $\gamma\in(0,1]$ at finite temperature, $J_{\mathrm{eff}}(\omega)\sim O(\omega^{\gamma-1}/\beta)$, so we cannot integrate by parts twice, only once. Following the same argument, we can show that $C_{\mathrm{error},\beta}(t) \leq\frac{1}{\beta} C_{\gamma, \omega_c}' \mathrm e^{-W/2\omega_c}\frac{1}{1+t}$. This means that to achieve an $L^1$ error on $[0,T]$ below $\varepsilon$, it suffices to choose $W = O(\log(\frac{\log(1+T)}{\beta\varepsilon}))$.
The final complexity scaling thus has an additional $\log\log\log(T)$ dependence for $\gamma\in(0,1]$ at finite temperature. Since this term is negligible compared with the dominant $\log T$ and $\log\log T$ dependences, it is omitted in \cref{tab:summary}.

\section{Classical generalized Langevin equations with spectral singularities}
\label{sec:classical}

We present two examples of classical generalized Langevin equations (GLEs) with spectral singularities that are widely used in the literature. Our results apply to these GLEs and thus provide rigorous error bounds for their numerical simulation.

The first example is to consider the classical lattice dynamics subject to a harmonic bath \cite{adelman1974generalized}.  In the 1D semi-infinite chain, we take the \emph{bath} atoms \(j\le 0\) to be harmonic with nearest-neighbor coupling (with spring stiffness \(K^2\)) while the interior atoms \(j\ge 1\) can remain anharmonic. Since the bath dynamics is linear, one can solve the equations of motion for the bath and substitute the result back into the boundary atom's equation, thereby eliminating all bath degrees of freedom in favor of a retarded friction (memory) term plus a fluctuating force determined by the bath initial conditions. This yields the exact generalized Langevin equation for the boundary atom \(j=1\):
\[
m\ddot u_1(t)=\phi'(u_2-u_1)-\int_0^t C(t-s)\,\dot u_1(s) \ud s+\xi(t),
\]
where $\phi$ is the anharmonic potential, $\xi(t)$ is the random force with $\langle \xi(t)\xi(0)\rangle = k_B T\,C(t),$ and the correlation kernel
\[
C(t)=\frac{\sqrt{m}\,K}{t}\,J_1\!\left(\omega_D t\right),\qquad \omega_D=\frac{2K}{\sqrt{m}},
\]
and, if one even-extends \(C(t)\) to \(\mathbb{R}\) via \(C_{\rm even}(t)=C(|t|)\), the corresponding (two-sided) spectral distribution
\begin{align*}
S(\omega)\;:=\;\int_{-\infty}^{\infty}C(|t|)\,e^{-i\omega t}\,dt
&=2\int_{0}^{\infty}C(t)\cos(\omega t)\,dt =m\sqrt{\omega_D^2-\omega^2}\;\mathbf{1}_{|\omega|<\omega_D},
\end{align*}
so \(S(\omega)\) has square-root nonanalyticity at the band edges \(\omega=\pm\omega_D\) (and vanishes for \(|\omega|>\omega_D\)). This is similar to the case of van Hove singularities in 3D (see \cref{fig:VanHove}(b)), so our results also apply to this GLE.

In practice, by approximating the memory kernel as a sum of exponentials, $C(t) \approx \sum_{j=1}^N c_j e^{-z_j t}$, the memory integral decomposes into $N$ independent auxiliary modes defined by $\phi_j(t) = \int_0^t e^{-z_j(t-s)} \dot{u}_1(s) \, ds$. Crucially, by differentiating with respect to time, one finds that each mode satisfies a local ordinary differential equation, $\dot{\phi}_j(t) = -z_j \phi_j(t) + \dot{u}_1(t)$. This transformation unwraps the convolution, allowing the full memory term to be computed as a weighted sum $\sum_{j=1}^N c_j \phi_j(t)$ without storing the system's history.  This is a direct classical analogue of the pseudomode approach. Crucially, the efficiency of this embedding relies entirely on the compactness of the sum, namely the number of modes $N$ required to achieve a given accuracy, which is precisely the quantity we bound in this work.

Another widely used GLE is the fractional generalized Langevin equation (FLE) for viscoelastic fluids:
\[
m\ddot x(t)+\int_{0}^{t}C(t-\tau)\,\dot x(\tau)\,d\tau+U'(x(t))=\xi(t),
\]
Here
$\langle \xi(t)\xi(t')\rangle = k_B T\,C(|t-t'|),$ 
i.e., the random force satisfies the equilibrium fluctuation--dissipation relation with the same memory kernel. A canonical fractional choice is a power-law kernel
\[
C(t)=\frac{\gamma_\nu}{\Gamma(1-\nu)}\,t^{-\nu}\quad (t>0,\;0<\nu<1),
\]
as discussed e.g. in  \cite{lutz2001fractional}. If one even-extends the kernel again, then its power spectrum has the closed form
\begin{align*}
S(\omega):= &
2\int_{0}^{\infty}C(t)\cos(\omega t)\,dt =2\gamma_\nu\,\sin\!\Big(\frac{\pi\nu}{2}\Big)\,|\omega|^{\nu-1}.
\end{align*}
Consequently, for \(0<\nu<1\), \(S(\omega)\sim |\omega|^{\nu-1}\) diverges as \(\omega\to 0\). In the notation of the main text, this corresponds to a spectral singularity of order $\alpha=\nu-1\in(-1,0)$ at $\omega=0$. Therefore the FLE falls into the strong-singularity regime in our complexity theory; in particular, our $L^1$ bounds allow $N$ to grow at most polylogarithmically in $T$, e.g. $N=O\!\left(\log^2\!\left(T/\varepsilon\right)\right)$ for fixed target accuracy.

\section{Numerical experiments in Fig. \ref{fig:n_vs_t}}
\label{sec:numerics}

Finally, we provide details for the numerical experiments in \cref{fig:n_vs_t}. Here we fit the BCF corresponding to the Ohmic density $J(\omega) = \omega \mathrm e^{-\omega} 1_{[0,+\infty)}$ at zero temperature. The BCF has the closed-form expression $\Delta(t) = \frac{1}{(1+\mathrm i t)^2}$. The fit uses the ESPRIT algorithm \cite{RoyKailath1989}. The time interval $[0,T]$ is discretized uniformly with step size $\Delta t = 0.01$, and the $L^1$ error is approximated by the Riemann sum $\sum_{k=0}^{T/\Delta t} |\Delta(k\Delta t) - \sum_{j=1}^N c_j e^{-\mathrm i z_j k \Delta t}| \Delta t$. The number of poles $N$ is chosen as the smallest integer such that the approximated $L^1$ error is below the target accuracy $\varepsilon$.

\end{document}